\PassOptionsToPackage{pagebackref}{hyperref}
\documentclass[a4paper,
               DIV=14,
               11pt,
               abstract=true,
               titlepage=off]{scrartcl}

\usepackage{amsmath,amssymb,amsthm,mathtools,mathrsfs,thm-restate}
\usepackage{microtype}
\usepackage[notes=true,later=true]{dtrt}
\usepackage{hyperref}
\usepackage{physics}
\usepackage{braket}
\usepackage{graphics,xcolor}
\usepackage[capitalize,nameinlink]{cleveref}
\usepackage{url}
\usepackage{authblk}

\makeatletter
\def\@footnotecolor{red}
\define@key{Hyp}{footnotecolor}{%
 \HyColor@HyperrefColor{#1}\@footnotecolor%
}
\def\@footnotemark{%
    \leavevmode
    \ifhmode\edef\@x@sf{\the\spacefactor}\nobreak\fi
    \stepcounter{Hfootnote}%
    \global\let\Hy@saved@currentHref\@currentHref
    \hyper@makecurrent{Hfootnote}%
    \global\let\Hy@footnote@currentHref\@currentHref
    \global\let\@currentHref\Hy@saved@currentHref
    \hyper@linkstart{footnote}{\Hy@footnote@currentHref}%
    \@makefnmark
    \hyper@linkend
    \ifhmode\spacefactor\@x@sf\fi
    \relax
  }%
\makeatother

\hypersetup{
  linkcolor  = blue!80!black,
  citecolor  = green!50!black,
  footnotecolor = red!90!black,
  colorlinks = true,
}

\newtheorem{introtheorem}{Theorem}
\newtheorem{theorem}{Theorem}[section]
\newtheorem{lemma}[theorem]{Lemma}
\newtheorem{corollary}[theorem]{Corollary}

\theoremstyle{definition}
\newtheorem{remark}[theorem]{Remark}
\newtheorem{definition}[theorem]{Definition}

\newcommand{\Field}{\mathbb F}
\newcommand{\Complex}{\mathbb C}
\newcommand{\dimension}{d} 
\newcommand{\qudits}{n}
\newcommand{\level}{k}
\newcommand{\Dimension}{D} 
\newcommand{\Group}{\mathcal G}
\newcommand{\AnotherGroup}{\mathcal H}
\newcommand{\Unitaries}{\mathcal{U}}
\newcommand{\RootOfUnity}{\omega}
\newcommand{\Paulis}{\mathcal P}
\newcommand{\Cliffords}{\mathcal C}
\newcommand{\Algebra}{\mathcal A}
\newcommand{\Lagrangians}{\mathscr L}
\newcommand{\Stabilizer}{\mathcal S}
\renewcommand{\Subset}{\mathcal{T}}
\newcommand{\ConjugationGroup}{\Gamma}
\newcommand{\MatrixSubspace}{\mathcal W}
\newcommand{\FixedLagrangians}{\mathscr F}

\newcommand{\lspan}{\operatorname{span}}

\newcommand{\PauliLabelSpace}{\Field_\dimension^{2\qudits}}

\title{The Generalized Semi-Clifford Conjecture Holds at Level 4}
\author[1]{Maxwell Marcus\thanks{\texttt{mm1568@princeton.edu}}}
\author[2]{Sathyawageeswar Subramanian\thanks{\texttt{Sathya.Subramanian@cs.ox.ac.uk}}}
\author[1]{Marcel Dall'Agnol\thanks{\texttt{dallagnol@cs.princeton.edu}}}
\affil[1]{\small \textit{Department of Computer Science, Princeton University, 35 Olden Street,
Princeton, NJ 08544
}}

\affil[2]{\small \textit{Department of Computer Science, University of Oxford, Parks Rd, Oxford OX1 3QG, United Kingdom}}
\date{}

\begin{document}
\maketitle

\begin{abstract}
The Clifford hierarchy $\mathcal{C}_1 \subset \mathcal{C}_2 \subset \cdots$ was introduced by Gottesman and Chuang (Nature 402, 1999) to characterize gates that admit fault-tolerant implementation by gate teleportation. Yet, despite its rich mathematical structure and the attention it has received in recent years, little is known about $\mathcal{C}_k$ for $k > 3$. Most progress has focused on identifying structural properties of restrictions of the hierarchy, such as diagonal gates and gates on systems of small dimension $d$ or with few qudits.

The \emph{generalized semi-Clifford conjecture}, proposed by Zeng et al.\ (Phys.\ Rev.\ A 77, 2008), states that every gate in $\cup_k \mathcal{C}_k$ is, up to multiplication by Cliffords, the product of a permutation and a diagonal matrix. Beigi and Shor proved the case $d=2$, $k=3$ (Quantum Inf.\ Comput.\ 10, 2010) and Pllaha et al.\ found an alternative proof by exploiting fixed points of the conjugation map induced by (a Clifford correction of) $U \in \mathcal{C}_3$ on the span of maximal stabilizer subgroups (Quantum 4, 2020). By extending their fixed-point arguments to the group $\ConjugationGroup_1(U)$ generated by $U \mathcal{P} U^\dagger$ and beyond, we prove the conjecture for $k \leq 4$ and any prime dimension $d$.

Our proof centers on \emph{conjugation groups} $\Gamma_1(U), \Gamma_2(U), \ldots$ of $U \in \mathcal{C}_k$, which we expect to be a useful tool in the study of the Clifford hierarchy more generally. We also show a natural sufficient condition on such groups for gates in higher levels to be generalized semi-Clifford. 
\end{abstract}

\clearpage
\tableofcontents
\clearpage

\section{Introduction}
Motivated by universal fault-tolerant quantum computing, Gottesman and Chuang \cite{GottesmanChuang1999} proposed \emph{gate teleportation} as a viable route towards this goal. In gate teleportation protocols, a gate is applied by consuming a pre-prepared resource state and correcting the measurement outcome with a gate drawn from one level lower in the nested sequence of unitaries
\[
\Cliffords_1 \subset \Cliffords_2 \subset \cdots \subset \Cliffords_\level \subset \cdots,
\]
known as the \emph{Clifford hierarchy}. Here $\Cliffords_1 = \Paulis$ is the Pauli group, $\Cliffords_2$ the Clifford group, $\Cliffords_\level$ consists of the unitaries that conjugate $\Paulis$ into $\Cliffords_{\level-1}$, and the level of a gate measures, roughly, the resources its teleportation protocol consumes \cite{ZhouLeungChuang2000}.

The Clifford hierarchy and its rich mathematical structure have since become objects of study in their own right. In the effort to better understand the structure of $\Cliffords_\level$ or restricted classes thereof (and characterize gates that require fewer resources to teleport), \cite{ZengChenChuang2008} defined \emph{semi-Clifford} and \emph{generalized semi-Clifford} gates. The former consist of those which conjugate some maximal abelian subgroup of $\Paulis$ onto another; and the latter, of those which map the \emph{span} of a maximal abelian subgroup (i.e., the matrix algebra it generates) onto the span of another.\footnote{These conditions are equivalent to semi-Cliffords being diagonal and generalized semi-Cliffords being the product of a permutation by a diagonal matrix, both up to left and right multiplication by $\Cliffords_2$.}

\paragraph{The (generalized) semi-Clifford conjecture.} \cite{ZengChenChuang2008} conjectured for all $(\qudits, \dimension, \level)$ that all $\qudits$-qudit unitaries over dimension-$\dimension$ qudits in $\Cliffords_\level$ are generalized semi-Clifford,\footnote{Strictly speaking, their conjecture was restricted to $\dimension = 2$, but naturally generalizes to all $\dimension$.} and semi-Clifford if $\level = 3$. They proved the entire hierarchy on up to two qubits ($\dimension = 2$,  $\qudits \leq 2$, all $\level$) is semi-Clifford, as is the third level on three qubits ($\dimension = 2$, $\qudits = \level = 3$), and observed that non-semi-Clifford gates exist with $\qudits > 2$ qubits at level $\level > 3$. A string of subsequent work \cite{deSilva2021,ChenDeSilva2024,deSilvaLautsch2025,AndersonConnelly2025} proved the semi-Clifford property for one prime-dimensional qudit ($\dimension$ prime, $\qudits = 1$, all $\level$) and, in the third level, up to two prime-dimensional qudits ($\dimension$ prime, $\qudits \leq 2$, $\level = 3$) as well as four qubits ($\dimension = 2$, $\qudits = 4$, $\level = 3$).

Non-semi-Clifford gates are known to exist in $\Cliffords_3$ (Gottesman and Mochon, reported in \cite{BeigiShor2010}) and recent work by de Silva and Lautsch also constructs non-generalized semi-Clifford gates at the fifth level \cite{desilva2026}. 
However, there have been no positive results on the generalized semi-Clifford conjecture (other than those implied by semi-Clifford characterizations) since Beigi and Shor's \cite{BeigiShor2010} proof of the case of third-level gates on qubits ($\dimension = 2$, $\level = 3$ and any $\qudits$). 

\paragraph{Our results.} In this paper, we settle the generalized semi-Clifford conjecture for fourth-level gates in prime dimension (all $\qudits$, $\dimension$ prime, $\level = 4$). Our ideas build on Pllaha, Rengaswamy, Tirkkonen, and Calderbank's \cite{PllahaEtAl2020} alternative proof for qubit gates with $\level = 3$ by analyzing the Pauli expansion of Clifford and $\Cliffords_3$ unitaries. Their analysis implies a fixed-point property of the conjugation of Paulis and subgroups thereof by (Clifford corrections of) $U \in \Cliffords_3$. Specifically, they show that, for some $C \in \Cliffords_2$, some $P \in \Paulis$ is fixed by conjugation under $CU$; this in turn implies that $CU$'s Pauli expansion is supported on a maximal commuting subgroup $\Stabilizer$ of $\Paulis$, and thus that $U$ conjugates the span of $\Stabilizer$ into the span of $C^\dagger \Stabilizer C$ (another maximal commuting subgroup).

We sidestep Pauli expansions and instead analyze the action of the \emph{conjugation groups} $\ConjugationGroup_1(U) \coloneqq U \Paulis U^\dagger \cong \Paulis$ (generated by the conjugate tuple corresponding to $U$, in the language of \cite{deSilva2021}) and $\ConjugationGroup_2(U) \coloneqq \langle V \Paulis V^\dagger : V \in U \Paulis U^\dagger\rangle$ over (subsets of) $\Lagrangians$, \emph{the set of Lagrangian subspaces of} $\PauliLabelSpace$ (equivalently, of maximal commuting subgroups of $\Paulis$). Showing that $\ConjugationGroup_2(U)$ is a $\dimension$-group and $\dimension$ does not divide $\abs{\Lagrangians}$ implies the action has a fixed point, and allows us to show that $U$ is generalized semi-Clifford.

\begin{introtheorem}
    \label{thm:level4}
    Let $\qudits \in \mathbb{N}$, $\level\leq 4$ and $\dimension$ be prime. Every $\qudits$-qudit gate in $\Cliffords_\level$ over qudits of dimension $\dimension$ is generalized semi-Clifford. 
\end{introtheorem}

Note that the second conjugation group readily generalizes to $j > 2$ via the recursive definition $\ConjugationGroup_j \coloneqq \langle V \Paulis V^\dagger : V \in \ConjugationGroup_{j - 1}\rangle$. We identify a nontrivial sufficient condition to show that a gate in an arbitrary level is generalized semi-Clifford, namely, that $\ConjugationGroup_{\level-2}(U) \subseteq \Cliffords_2$; this implies $\ConjugationGroup_{\level - j} \subseteq \Cliffords_j$ is a $\dimension$-group for all $j$ and that the conjugation groups form a normal chain $\ConjugationGroup_{\level - 2}(U) \trianglelefteq \ConjugationGroup_{\level - 3}(U)\trianglelefteq \cdots \trianglelefteq \ConjugationGroup_2(U)$, with which our proof generalizes (see \cref{thm:conditional-higher-gsc}).

Our definition of conjugation groups is the key conceptual contribution of this work, and we expect them to be a useful addition to the tools currently available for the analysis of the Clifford hierarchy.

\paragraph{Related work.}
While there has been little work on the generalized semi-Clifford conjecture itself, properties of the Clifford hierarchy have been studied extensively.

For $\Cliffords_3$, \cite{FengLuo2024} determined the group structure of diagonal gates and \cite{HeRobitailleTan2024,HeRobitailleTan2025} characterized the permutation gates. \cite{AndersonWeippert2024} derived necessary conditions for controlled gates to lie in the hierarchy, while \cite{XuWang2026} determined the ascent within the hierarchy caused by adding a control qubit to (a family of) $\Cliffords_2$ gates.

\cite{Anderson2024,BengtssonEtAl2014,BastioniEtAl2026} studied, respectively, the (generalized) semi-Clifford groups, the order-three symmetries, and the square roots (of Hermitians) contained in the hierarchy. Despite this substantial body of work, to the best of our knowledge, the objects we study have not yet been investigated systematically.

\section{Preliminaries}

\paragraph{Notation.}
Throughout the paper, $\dimension$ denotes a prime number corresponding to the local dimension of a quantum system (i.e., qudits of dimension $\dimension$) and $\qudits$ denotes the number of qudits in the system. The total dimension is $\Dimension=\dimension^\qudits$, and all unitaries are transformations over $(\Complex^{\dimension})^{\otimes \qudits} \cong \Complex^{\Dimension}$. We also use $\Field_\dimension$ to denote the finite field with $\dimension$ elements, and $\RootOfUnity$ to denote the primitive $\dimension^\text{th}$ root of unity $e^{2 \pi i/\dimension}$. The computational basis of $\Complex^\dimension$ is $\set{\ket{x} : x \in \Field_\dimension}$. 

We usually consider subgroups and subsets of unitaries quotiented by phase; we denote the quotient group with calligraphic letters (e.g., $\Paulis = \Cliffords_1$ and $\Cliffords_2$ denote the Pauli and Clifford groups modulo phase respectively) and explicitly multiply by phases when required (e.g., $\langle \RootOfUnity \rangle \Paulis$ is the Pauli group with phase).\footnote{Note that we do not consider, e.g., $\mathbb{S}_1 \Paulis = \set{e^{i\theta} P : P \in \Paulis}$, the Pauli group with all complex phases.} With slight abuse of notation, we also write $\Subset$ for a subset (that is not necessarily a subgroup) of $\Unitaries$, the unitary group modulo phase.

We use $\Group \leq \AnotherGroup$ (respectively, $\Group \trianglelefteq \AnotherGroup$) to mean that $\Group$ is a subgroup (respectively, normal subgroup) of $\AnotherGroup$. $\langle U \rangle$ and $\langle \Subset \rangle$ denote the multiplicative groups generated by $U$ and by $\Subset$, respectively; and we also write $U \Subset V$ for $\set{U S V : S \in \Subset}$.  

\subsection{Symplectic spaces}
We first define the group of $\qudits$-qudit Pauli operators, which plays a fundamental role in quantum physics.
\begin{definition}
The Pauli group $\Paulis \coloneqq \langle X, Z\rangle^{\otimes \qudits} / \langle \RootOfUnity I \rangle$ is generated by the unitaries $X$ and $Z$, defined by the mappings $\ket{x} \mapsto \ket{x + 1 \pmod \dimension}$ and $\ket{x} \mapsto \RootOfUnity^{x} \ket{x}$, respectively.
\end{definition}

Elements of $\Paulis$ can be represented by the label space $\PauliLabelSpace$ equipped with the symplectic form $[(x_1, x_2),(y_1, y_2)] \coloneqq x_1 \cdot y_2 - x_2 \cdot y_1$.
Then $\langle \RootOfUnity \rangle \Paulis=\set{\RootOfUnity^i P_x:x\in\PauliLabelSpace, i \in [\dimension]}$ satisfies $P_x P_y \propto P_{x+y}$ and $P_xP_yP_x^\dagger=\RootOfUnity^{[x,y]}P_y$. Under this correspondence, subgroups of $\Paulis$ are related to subspaces of $\PauliLabelSpace$. A distinguished family of subspaces that we will work with are called Lagrangians.
\begin{definition}
A subspace $L\subset \PauliLabelSpace$ is \emph{isotropic} if $[v,w]=0$  for all $v,w\in L$. If $L$ is maximal among isotropic subspaces, it is \emph{Lagrangian}.
\end{definition}

Note that Lagrangians are equivalently defined as isotropic subspaces of dimension $n$: isotropy gives $L\subseteq L^\perp = \set{y : [x, y] = 0 \text{ for all } x \in L}$, maximality gives $L=L^\perp$, and non-degeneracy gives $\dim L + \dim L^\perp=2n$; and an isotropic subspace $L$ of dimension $n$ is a Lagrangian since $L\subseteq L^\perp$ and $\dim L + \dim L^\perp=2n$.

Maximal commuting subgroups of $\langle \RootOfUnity \rangle \Paulis$ (more precisely, equivalence classes of such subgroups modulo phase) are in bijective correspondence with Lagrangians in $\PauliLabelSpace$ via the mapping $L \mapsto \Stabilizer_L = \set{P_x:x\in L}$. The isotropy of $L$ implies $\Stabilizer_L$ is abelian, and maximality of $\Stabilizer_L$ follows from that of $L$. Since Lagrangians have dimension $\qudits$, they and their associated maximal commuting subgroups have cardinality $\Dimension = \dimension^\qudits$.

Composing the above correspondence for $L$ with the natural mapping from multiplicative subgroups of the matrix algebra $M_{\Dimension}(\Complex)$ to subalgebras yields $\Algebra_L = \lspan_\mathbb{C} \Stabilizer_L$. These are clearly abelian, maximal (among abelian subalgebras) and spanned by Paulis.

\begin{definition}
    A \emph{Pauli Maximal Abelian Subalgebra (Pauli MASA)} $\Algebra$ is a maximal abelian subalgebra of $M_{\Dimension}(\Complex)$ such that $\Algebra = \lspan_\Complex \Subset$ with $\Subset \subseteq \Paulis$. Pauli MASAs associated to a Lagrangian $L$ are denoted $\Algebra_L \coloneqq \lspan_\Complex\set{P_x : x \in L}$.
\end{definition}

Note that conjugation by any unitary takes a MASA to a MASA and preserves its dimension. But not all maximal abelian subalgebras are spanned by Paulis; indeed, conjugation by an arbitrary unitary maps a Pauli MASA to a MASA that is not necessarily Pauli.

We note two further facts about Lagrangian subspaces and MASAs (\cref{cor:nonempty-fixed-point-set,cor:pauli-masa}) that we will use later.

\begin{lemma}[\cite{Gross2006}]
\label{lem:lagrangian-subspaces}
The number of Lagrangian subspaces of $\PauliLabelSpace$ is
$\prod_{j=1}^\qudits(\dimension^j+1)$.
\end{lemma}
\noindent We defer the proof of this fact (an immediate corollary of the formula for the number of isotropic subspaces of each dimension) to \cref{sec:deferred-proofs}.

\begin{definition}
    For a prime number $p$, a group $G$ is a $p$-group if every element $g\in G$ has order a power of $p$. If $G$ is a finite group, $|G|=p^n$ for some $n$ is an equivalent condition. In this paper, we only consider $\dimension$-groups, where $\dimension$ is the qudit dimension.
\end{definition}

\begin{lemma}
    \label{lem:d-group-action-size}
    If a finite $\dimension$-group $\Group$ acts on a finite set $\Subset$, then $\abs{\Subset^\Group} \equiv \abs{\Subset} \pmod{\dimension}$, where $\Subset^\Group$ is the fixed-point set $\Subset^\Group = \set{T \in \Subset: G \cdot T = T \text{ for all } G \in \Group}$ of the action.
\end{lemma}
\begin{proof}
    By the orbit-stabilizer theorem, the size of every orbit $\abs{\Group \cdot T}$ divides $\abs{\Group}$ and is hence a power of $\dimension$. Orbits of size $1$ are exactly the elements of $\Subset^\Group$, and every other orbit has size divisible by $\dimension$. Decomposing $\Subset$ into $\Group$-orbits, we have $\abs{\Subset^\Group} \equiv \abs{\Subset} \pmod{\dimension}$.
\end{proof}

\begin{corollary}
    \label{cor:nonempty-fixed-point-set}
    If a finite $\dimension$-group $\AnotherGroup$ acts on the set $\Lagrangians$ of Lagrangian subspaces of $\Field_\dimension^{2\qudits}$, then the fixed-point set $\FixedLagrangians = \Lagrangians^\AnotherGroup$ satisfies $\abs{\FixedLagrangians} \equiv 1 \pmod{\dimension}$. Moreover, if a finite $\dimension$-group $\Group$ acts on $\FixedLagrangians$, then $\abs{\FixedLagrangians^\Group} \equiv 1 \pmod{\dimension}$ as well; in particular, $\FixedLagrangians \neq \varnothing$ and $ \FixedLagrangians^\Group \neq \varnothing$.
\end{corollary}
\begin{proof}
    The claim follows from $\abs{\Lagrangians} = \prod_{j=1}^{\qudits}(\dimension^j+1) \equiv 1 \pmod{\dimension}$ (\cref{lem:lagrangian-subspaces}) and \cref{lem:d-group-action-size} (applied twice: once for $\AnotherGroup$ and $\Lagrangians$, and again for $\Group$ and $\FixedLagrangians$). 
\end{proof}

The following lemma and corollary will allow us to show under certain circumstances that an abelian subalgebra that we get from conjugating a Pauli MASA is also a Pauli MASA. 
\begin{lemma}[Pauli invariance]
\label{lem:pauli-invariance}
Let $\MatrixSubspace \subseteq M_{\Dimension
}(\Complex)$ be a linear subspace satisfying
\[
P_x \MatrixSubspace P_x^\dagger=\MatrixSubspace
\]
for all $x \in \Field_\dimension^{2\qudits}$. Then $\MatrixSubspace = \lspan_\Complex\Subset$ with $\Subset \subseteq \Paulis$. 
\end{lemma}
We defer the proof of \cref{lem:pauli-invariance} to \cref{sec:deferred-proofs}. As its immediate corollary, we have:

\begin{corollary}
\label{cor:pauli-masa}
If $\Algebra \subseteq M_\Dimension(\Complex)$ is an abelian subalgebra of dimension
$\Dimension$ such that $P_x \Algebra P_x^\dagger= \Algebra$ for every $x\in\PauliLabelSpace$, then $\Algebra$ is a Pauli MASA.
\end{corollary}
\begin{proof}
    \cref{lem:pauli-invariance} implies existence of a Pauli basis $\Stabilizer \subseteq \Paulis$ such that $\Algebra = \lspan_\Complex \Stabilizer$; define $L \coloneqq \set{x : P_x \in \Stabilizer}$. $\Algebra$ abelian implies $[x,y] = 0$ for all $x, y \in L$, and thus that $\lspan_{\Field_\dimension} L$ is an isotropic subspace (of size at most $\Dimension$). Finally, $\abs{L} = \abs{\Stabilizer} = \Dimension$ shows that $L = \lspan_{\Field_\dimension} L$ is a Lagrangian subspace. Then $\Algebra = \Algebra_L$, which in particular is maximal.
\end{proof}

\subsection{The Clifford hierarchy}
\begin{definition}
    The Clifford hierarchy $\Cliffords_1 \subset \Cliffords_2 \subset \cdots \subset \Cliffords_\level \subset \cdots$ is defined inductively. 
    Let 
    \[\Cliffords_1 \coloneqq \Paulis\] and
    \[\Cliffords_{k}\coloneqq\set{U\in \Unitaries: UP_xU^{\dagger}\in \Cliffords_{k-1} \text{ for all } x\in \PauliLabelSpace}.\]
\end{definition}
\noindent Note that $\Cliffords_2$ (the normalizer of $\Paulis$) is the Clifford group, but $\Cliffords_\level$ is not a group for $\level > 2$.

\begin{definition}[\cite{ZengChenChuang2008}]
    A unitary $U$ is \emph{generalized semi-Clifford} if there exist Lagrangians $L$ and $S$ such that $U\Algebra_LU^\dagger=\Algebra_S$.
\end{definition}
Note that conjugating some Pauli MASA to another is equivalent to admitting the decomposition $U = C R D C'$, where $C, C' \in \Cliffords_2$, $R$ is a permutation matrix and $D \in \Cliffords_\level$ is diagonal \cite{ZengChenChuang2008}; we will not make use of this fact, however.

\section{The core argument at level 3}
In this section, we reprove the fact that the third level of the Clifford hierarchy is generalized semi-Clifford. Although this result has long been known in the qubit case \cite{BeigiShor2010}, our proof generalizes to prime qudits and isolates the key ideas that extend to $\Cliffords_4$. 

In their proof that $\Cliffords_3$ (for qubits) is generalized semi-Clifford, \cite{PllahaEtAl2020} show that every gate $U \in \Cliffords_3$ admits a Clifford correction $C \in \Cliffords_2$ such that $CU$ fixes some Pauli under conjugation; then, they show this implies $CU$ also fixes the span of a maximal commuting subgroup (MCS), and thus that $U$ conjugates the span of an MCS into that of another MCS. Our proof retains this fixed-point idea but develops it in two directions. First, we consider the action of the first conjugation group on the set of algebras associated to Lagrangians (equivalently, spans of MCSs) directly. Second, we generalize the qubit-specific argument into one that holds in arbitrary prime dimension $\dimension$. 

As \cite{PllahaEtAl2020}, we investigate the group $U\Paulis U^\dagger \cong \Paulis$, which we define explicitly as the first conjugation group and generalize in \cref{def:conjugation-groups}. The generalization of this construction is the conceptual advance that will allow us to prove \cref{thm:level4}.

\begin{theorem}
Every gate on $n$ qudits of prime dimension $\dimension$ in the third level of the Clifford hierarchy is
generalized semi-Clifford.
\end{theorem}
\begin{proof}
Let $U\in \Cliffords_3$. Define
\[
\Group:= U \Paulis U^\dagger.
\]
Since $\Group\cong \Paulis$ and $\abs{\Paulis} = \dimension^{2\qudits}$, in particular $\Group$ is a $\dimension$-group.

Let $\Lagrangians$ be the set of Lagrangian subspaces of $\PauliLabelSpace$, and $\Algebra_L$ be the Pauli MASA associated to $L \in \Lagrangians$. Since for every $L \in \Lagrangians$ and $C \in \Cliffords_2$ there exists $S \in \Lagrangians$ such that $C\Algebra_LC^\dagger=\Algebra_S$ (namely, the MASA generated by $C P_x C^\dagger \in \Paulis$ for all $x \in L$) the group $\Group \leq \Cliffords_2$ acts by conjugation on $\Lagrangians$.

Let
\[
\FixedLagrangians\coloneqq \Lagrangians^\Group = 
\set{L\in\Lagrangians:G\Algebra_LG^\dagger=\Algebra_L
\text{ for every }G\in \Group}
\]
be the fixed-point set of the action. By \cref{cor:nonempty-fixed-point-set}, $\FixedLagrangians \neq \varnothing$. Then, take $L\in \FixedLagrangians$ and set
\[
\Algebra:=U^\dagger\Algebra_LU.
\]
$\Algebra$ is an abelian subalgebra of dimension $\Dimension$ where, for every $P_x \in \Paulis$ with associated $G_x = U P_x U^\dagger \in \Group$,
\[
\begin{aligned}
P_x\Algebra P_x^\dagger
&=P_xU^\dagger \Algebra_L U P_x^\dagger\\
&=U^\dagger (UP_xU^\dagger) \Algebra_L (UP_x^\dagger U^\dagger)U\\
&=U^\dagger G_x\Algebra_LG_x^\dagger U\\
&=U^\dagger\Algebra_LU\\
&=\Algebra
\end{aligned}
\]
Finally, $\Algebra$ is a Pauli MASA by \cref{cor:pauli-masa}, i.e., $\Algebra = \Algebra_S$ for some $S\in \Lagrangians$, so
$U\Algebra_SU^\dagger=\Algebra_L$. We thus conclude that $U$ is generalized semi-Clifford. 
\end{proof}

\section{\texorpdfstring{$\Cliffords_4$}{C4} is generalized semi-Clifford}

We now present the proof of our main theorem. We first define conjugation groups and investigate some of their properties, which will set us up for the actual proof.

\begin{definition}
    \label{def:conjugation-groups}
    We define the \emph{conjugation groups} $\ConjugationGroup_j(U)$ for $j \in \mathbb{N}$ inductively.
    Let $U\in \Cliffords_\level$. Define $\ConjugationGroup_1(U) \coloneqq U \Paulis U^\dagger = \set{U P_x U^\dagger: x\in \PauliLabelSpace}$ and, for $j > 1$, define
    \[
    \ConjugationGroup_j(U) \coloneqq \left\langle V P_x V^\dagger: V \in \ConjugationGroup_{j-1}(U), x\in \PauliLabelSpace \right\rangle.
    \]
\end{definition}
\noindent When $U$ is fixed, we write $G_x:=UP_xU^\dagger$ and 
$H_{x,y}:=G_xP_yG_x^\dagger$, so that
\[
\ConjugationGroup_1(U)=\langle G_x:x\in\PauliLabelSpace\rangle = \set{G_x:x\in\PauliLabelSpace} \qquad\text{and}
\qquad
\ConjugationGroup_2(U)=\langle H_{x,y}:x,y\in\PauliLabelSpace\rangle.
\]
Note, moreover, that $\ConjugationGroup_1(U)\subset \Cliffords_{\level-1}$ by definition, but the containment is not necessarily true for other conjugation groups. For $3 \leq \level \leq 4$, however, the generators of $\ConjugationGroup_2(U)$ lie in $\Cliffords_{\level-2}$, which is a group; the inclusion $\ConjugationGroup_2(U)\subset \Cliffords_{\level-2}$ is then immediate.

\subsection{\texorpdfstring{$\ConjugationGroup_2$}{Y2} is a \texorpdfstring{$\dimension$}{d}-group}
The first and second conjugation groups will be the main tools for finding the Pauli MASA that $U\in \Cliffords_4$ sends to another Pauli MASA. Here we show that the second conjugation group is a $\dimension$-group, which will allow us to find a fixed point of the action on the set of Lagrangians, as in the proof that $\Cliffords_3$ is generalized semi-Clifford. 

\begin{definition}
The $\dimension$-core of a group $\AnotherGroup$, denoted $O_\dimension(\AnotherGroup)$, is the maximal normal
$\dimension$-subgroup of $\AnotherGroup$.
\end{definition}
Note that $O_\dimension(\AnotherGroup)$ contains every normal $\dimension$-subgroup of $\AnotherGroup$: if $M$ is a maximal normal $\dimension$-subgroup and $N$ is any normal
$\dimension$-subgroup, then $MN$ is a normal $\dimension$-subgroup, so $MN=M$ by maximality. In particular, $O_\dimension(\AnotherGroup)$ is unique.

\begin{lemma}
$O_\dimension(\AnotherGroup)$ is characteristic in $\AnotherGroup$, i.e., $\phi(O_\dimension(\AnotherGroup))=O_\dimension(\AnotherGroup)$ for any automorphism $\phi$.
\end{lemma}
\begin{proof}
Let $\phi\in \operatorname{Aut}(\AnotherGroup)$. The images of $O_\dimension(\AnotherGroup)$ under
$\phi$ and $\phi^{-1}$ are normal $\dimension$-groups, so $\phi(O_\dimension(\AnotherGroup)) \subseteq O_\dimension(\AnotherGroup)$ and $\phi^{-1}(O_\dimension(\AnotherGroup)) \subseteq O_\dimension(\AnotherGroup)
\implies O_\dimension(\AnotherGroup)\subseteq\phi(O_\dimension(\AnotherGroup))$. Therefore, $\phi(O_\dimension(\AnotherGroup))=O_\dimension(\AnotherGroup)$.
\end{proof}
For the remainder of this section, fix $U\in\Cliffords_4$ and write
\[
\AnotherGroup:=\ConjugationGroup_2(U)= \left\langle H_{x,y}:x,y\in\PauliLabelSpace\right\rangle.
\]
The idea now is to show the generators $G_x \Paulis G_x^\dagger$ of $\AnotherGroup$ are in the $\dimension$-core. We first show $G_x$ induces an automorphism and then show the Paulis are in the $\dimension$-core. 
\begin{lemma}
\label[lemma]{lem:Qa-automorphism}
    For each $x\in \PauliLabelSpace$, conjugation by $G_x = U P_x U^\dagger$ induces an automorphism on $\AnotherGroup$.
\end{lemma}
\begin{proof}
Conjugation by $G_x = U P_x U^\dagger$ permutes the generating set of $\AnotherGroup$:
\begin{align*}
G_xH_{y,z}G_x^{\dagger}
&=G_xG_yP_z(G_x G_y)^\dagger\\
&= G_{x+y}P_zG_{x+y}^\dagger\\
&=H_{x+y,z}.\qedhere
\end{align*}
\end{proof}
\begin{theorem}
\label{thm:H_d-group}
For $U\in \Cliffords_4$, the second conjugation group
\[
\AnotherGroup = \ConjugationGroup_2(U) = \left\langle H_{x, y} = (UP_xU^\dagger) P_y(UP_xU^\dagger)^\dagger:
x,y\in\PauliLabelSpace\right\rangle\subseteq\Cliffords_2
\]
is a finite $\dimension$-group.
\end{theorem}

\begin{proof}
First, observe that $\Paulis \subseteq \AnotherGroup$, because $H_{0,y}=P_y$ for all $y$. Since $\AnotherGroup\leq\Cliffords_2$ and $\Paulis\trianglelefteq \Cliffords_2$, it follows that $\Paulis\trianglelefteq \AnotherGroup$.

Let $O_\dimension(\AnotherGroup)$ be the $\dimension$-core of $\AnotherGroup$. Since the projective Pauli group
is a normal $\dimension$-subgroup, $\Paulis\le O_\dimension(\AnotherGroup)$. Moreover, since $O_\dimension(\AnotherGroup)$ is characteristic in $\AnotherGroup$, we have \[G_xO_\dimension(\AnotherGroup)G_x^\dagger=O_\dimension(\AnotherGroup).\]

Finally, $\Paulis \subset O_\dimension(\AnotherGroup)$ implies 
\[
H_{x,y}=G_xP_yG_x^\dagger\in O_\dimension(\AnotherGroup)
\]
for every $x,y\in \PauliLabelSpace$. As $\AnotherGroup = \langle H_{x, y} : x, y \in \PauliLabelSpace\rangle$,
\[
\AnotherGroup\le O_\dimension(\AnotherGroup);
\]
the reverse inclusion is automatic, so $\AnotherGroup=O_\dimension(\AnotherGroup)$ and therefore $\AnotherGroup$ is a $\dimension$-group.
\end{proof}

\subsection{Proof of \texorpdfstring{\cref{thm:level4}}{Theorem \ref{thm:level4}}}
We now prove our main theorem.

\begin{theorem}[\cref{thm:level4}, restated]
Let $\qudits \in \mathbb{N}$ and $\level\leq 4$ and $\dimension$ be prime. Every $\qudits$-qudit gate $U\in\Cliffords_\level$ with qudits of dimension $\dimension$ is generalized semi-Clifford. 
\end{theorem}

\begin{proof}
Let $U\in \Cliffords_4$.  Denote the first conjugation group 
\[
\Group:=\ConjugationGroup_1(U) = \set{G_x=UP_xU^\dagger: x\in \PauliLabelSpace}
\]
and the second conjugation group
\[
\AnotherGroup := \ConjugationGroup_2(U) = \left\langle H_{x,y}=G_xP_yG_x^{\dagger}: x,y\in \PauliLabelSpace\right\rangle.
\]
$\Group \cong \Paulis$ is a $\dimension$-group since $\abs{\Group}=\abs{\Paulis}=d^{2n}$, and $\AnotherGroup$ is a
$\dimension$-group by \cref{thm:H_d-group}.

Let $\Lagrangians$ be the set of Lagrangian subspaces of $\PauliLabelSpace$, and denote the Pauli MASA generated by $L\in \Lagrangians$ by 
\[
\Algebra_L=\lspan\set{P_x: x\in L}.
\] 
For any $L\in \Lagrangians$ and any $C\in \Cliffords_2$,
\[
C\Algebra_LC^\dagger=\Algebra_S \quad\text{for some } S\in \Lagrangians. 
\]
Thus, conjugation of $\Algebra_L$ induces an action of the $\dimension$-group
$\AnotherGroup\subseteq \Cliffords_2$ on $\Lagrangians$. Let
\[
\FixedLagrangians:=
\set{L\in\Lagrangians:H\Algebra_LH^\dagger=\Algebra_L
\text{ for every }H\in \AnotherGroup}
\]
be the fixed-point set of the action. By \cref{cor:nonempty-fixed-point-set}, $\abs{\FixedLagrangians} \equiv 1 \pmod{\dimension}$, and thus $\FixedLagrangians \neq \varnothing$.

Now we claim conjugation of $\Algebra_L$ by $G_x\in \Group$ induces an action of $\Group$ on $\FixedLagrangians$. Fix $L\in\FixedLagrangians$ and set
\[
\Algebra_x:=G_x\Algebra_LG_x^\dagger.
\]
Conjugation by $G_x$ is an action as long as, for all $x \in \PauliLabelSpace$, $\Algebra_x = \Algebra_{L_x}$ for some $L_x \in \FixedLagrangians$ (i.e., $\Algebra_x$ is a Pauli MASA fixed by $\AnotherGroup$-conjugation); we next prove this fact.

First, we show $\Algebra_x$ is a Pauli MASA. For every Pauli $P_y$, using
$G_x^\dagger\propto G_{-x}$ and hence
$G_x^\dagger P_yG_x \propto H_{-x,y}\in \AnotherGroup$, we have
\[
\begin{aligned}
P_y\Algebra_xP_y^\dagger
&=G_x(G_x^\dagger P_yG_x)\Algebra_L
   (G_x^\dagger P_y^\dagger G_x)G_x^\dagger\\
&=G_xH_{-x,y}\Algebra_LH_{-x,y}^\dagger G_x^\dagger\\
&=G_x\Algebra_LG_x^\dagger \\
&=\Algebra_x,
\end{aligned}
\]
where the second-to-last equality is due to $L\in \FixedLagrangians$. 
Since $\Algebra_x$ is an abelian subalgebra of dimension $\Dimension$, \cref{cor:pauli-masa} implies that $\Algebra_x$ is a Pauli MASA.
That is, there is a unique $L_x\in\Lagrangians$ such that
\[
\Algebra_x=\Algebra_{L_x}.
\]

Now we show $L_x \in \FixedLagrangians$, i.e. $\Algebra_x$ is fixed by $\AnotherGroup$. For any $H\in \AnotherGroup$,
\[
\begin{aligned}
H\Algebra_{L_x}H^\dagger
&=H(G_x\Algebra_LG_x^\dagger)H^\dagger\\
&=G_x(G_x^\dagger HG_x)\Algebra_L (G_x^\dagger H^\dagger G_x)G_x^\dagger\\
&=G_x \Algebra_L G_x^\dagger\\
&=\Algebra_{L_x},
\end{aligned}
\]
where the second-to-last equality follows from \cref{lem:Qa-automorphism}, which yields $G_x^\dagger HG_x\in \AnotherGroup$ (and from $\Algebra_L$ being fixed by $\AnotherGroup$ due to $L \in \FixedLagrangians$). 
Therefore, $L_x\in\FixedLagrangians$.

Since the $\dimension$-group $\Group$ acts on $\FixedLagrangians$, \cref{cor:nonempty-fixed-point-set} yields a fixed point of the action, i.e., $S\in\FixedLagrangians$ such that
\[
G_x\Algebra_SG_x^\dagger=\Algebra_S
\qquad\text{for every }x\in \PauliLabelSpace.
\]
Finally, set
\[
\Algebra:=U^\dagger\Algebra_SU.
\]
For every Pauli $P_x$,
\[
\begin{aligned}
P_x\Algebra P_x^\dagger
&=P_xU^\dagger \Algebra_S U P_x^\dagger\\
&=U^\dagger (UP_xU^\dagger) \Algebra_S (UP_x^\dagger U^\dagger)U\\
&=U^\dagger G_x\Algebra_SG_x^\dagger U\\
&=U^\dagger\Algebra_SU\\
&=\Algebra.
\end{aligned}
\]
By \cref{cor:pauli-masa}, $\Algebra$ is a Pauli MASA, so $\Algebra=\Algebra_T$ for some $T\in \Lagrangians$, and
\[
U\Algebra_TU^\dagger=\Algebra_S.
\]
Therefore, $U$ is generalized semi-Clifford. 

\end{proof}

\section{Beyond Level \texorpdfstring{$4$}{4}}
Recent work by de Silva and Lautsch constructs a counterexample to the generalized semi-Clifford conjecture at level $\level = 5$ \cite{desilva2026}. Nevertheless, understanding which gates fall into this category and which do not remains useful. Here we give a sufficient condition for any gate to be generalized semi-Clifford. The main idea of this condition is to conceptually extend the proof for $\Cliffords_4$ to higher levels of the hierarchy. We do so by building a chain of $\dimension$-groups, namely the conjugation groups, such that each group normalizes the following group. Thus we can take advantage of the fixed-point properties of $\dimension$-groups and find a fixed Pauli MASA.

Let $U\in\Cliffords_\level$, write $G_x:=UP_xU^\dagger$, and abbreviate the conjugation groups by
\[
\ConjugationGroup_j:=\ConjugationGroup_j(U).
\]
Since $\ConjugationGroup_1=U\Paulis U^\dagger\cong\Paulis$, $\ConjugationGroup_1$
is a finite $\dimension$-group.

\begin{lemma}
\label{lem:successive-normalization}
For every $2\leq j\leq \level-2$, conjugation by every $V\in\ConjugationGroup_{j-1}$ induces an automorphism of  $\ConjugationGroup_j$.
\end{lemma}

\begin{proof}
Conjugation by $V\in\ConjugationGroup_{j-1}$ permutes the generating set of $\ConjugationGroup_j$:
\[
V(WP_yW^\dagger)V^\dagger
=(VW)P_y(VW)^\dagger\in\ConjugationGroup_j
\]
for every $W\in\ConjugationGroup_{j-1}$ and $y\in\PauliLabelSpace$. 
Applying the same argument to $V^\dagger$ proves equality, so conjugation by $V$ induces an automorphism of $\ConjugationGroup_j$.
\end{proof}

\begin{lemma}
\label{lem:conjugation-groups-normal}
For every $2\leq j\leq \level-3$,
\[
\ConjugationGroup_{j+1}\trianglelefteq\ConjugationGroup_j.
\]
\end{lemma}

\begin{proof}
Since $I\in\ConjugationGroup_{j-1}$, the definition of $\ConjugationGroup_j$ gives
\[
\Paulis\leq\ConjugationGroup_j.
\]
Consequently, for $V\in\ConjugationGroup_j$ and $y\in\PauliLabelSpace$,
\[
VP_yV^\dagger\in\ConjugationGroup_j,
\]
and hence $\ConjugationGroup_{j+1}\leq\ConjugationGroup_j$.

Moreover, conjugation by $W\in\ConjugationGroup_j$ permutes the generators of
$\ConjugationGroup_{j+1}$:
\[
W(VP_yV^\dagger)W^\dagger
=(WV)P_y(WV)^\dagger\in\ConjugationGroup_{j+1}.
\]
Thus $\ConjugationGroup_{j+1}\trianglelefteq\ConjugationGroup_j$.
\end{proof}

\begin{theorem}
\label{thm:all-conjugation-groups-d-groups}
Let $U\in\Cliffords_\level$, where $\level\geq4$. If
\[
\ConjugationGroup_{\level-2}\subseteq\Cliffords_2,
\]
then, for every $1\leq j\leq \level-2$,
\[
\ConjugationGroup_j\subseteq\Cliffords_{\level-j},
\]
and $\ConjugationGroup_j$ is a finite $\dimension$-group.
\end{theorem}

\begin{proof}
We first establish the hierarchy inclusions. The base case
\[
\ConjugationGroup_{\level-2}\subseteq\Cliffords_2
\]
holds by assumption. Suppose
\[
\ConjugationGroup_{j+1}\subseteq\Cliffords_{\level-j-1}.
\]
For every $V\in\ConjugationGroup_j$ and $y\in\PauliLabelSpace$, the definition of
$\ConjugationGroup_{j+1}$ gives
\[
VP_yV^\dagger\in\ConjugationGroup_{j+1}
\subseteq\Cliffords_{\level-j-1}.
\]
Therefore
\[
V\in\Cliffords_{\level-j},
\]
and hence
\[
\ConjugationGroup_j\subseteq\Cliffords_{\level-j}.
\]
Descending induction gives
\[
\ConjugationGroup_j\subseteq\Cliffords_{\level-j}
\qquad
1\leq j\leq \level-2.
\]
Since every fixed level of the (projective) Clifford hierarchy is finite, every $\ConjugationGroup_j$ is finite.

We now prove that the groups are $\dimension$-groups. First consider
$\ConjugationGroup_{\level-2}$. Since
\[
\ConjugationGroup_{\level-2}\leq\Cliffords_2
\qquad\text{and}\qquad
\Paulis\trianglelefteq\Cliffords_2,
\]
we have
\[
\Paulis\trianglelefteq\ConjugationGroup_{\level-2}.
\]
Let $O_\dimension(\ConjugationGroup_{\level-2})$ be its $\dimension$-core. Since $\Paulis$ is a normal
$\dimension$-subgroup,
\[
\Paulis\leq O_\dimension(\ConjugationGroup_{\level-2}).
\]

By \cref{lem:successive-normalization}, conjugation by every
$V\in\ConjugationGroup_{\level-3}$ induces an automorphism of $\ConjugationGroup_{\level-2}$.
Since the $\dimension$-core is characteristic,
\[
VO_\dimension(\ConjugationGroup_{\level-2})V^\dagger
=
O_\dimension(\ConjugationGroup_{\level-2}).
\]
Thus
\[
VP_yV^\dagger\in O_\dimension(\ConjugationGroup_{\level-2})
\]
for every $V\in\ConjugationGroup_{\level-3}$ and $y\in\PauliLabelSpace$. These elements generate
$\ConjugationGroup_{\level-2}$, so
\[
\ConjugationGroup_{\level-2}\leq O_\dimension(\ConjugationGroup_{\level-2}).
\]
Therefore
\[
\ConjugationGroup_{\level-2}=O_\dimension(\ConjugationGroup_{\level-2}),
\]
and $\ConjugationGroup_{\level-2}$ is a $\dimension$-group.

Now suppose $2\leq j\leq \level-3$ and that $\ConjugationGroup_{j+1}$ is a
$\dimension$-group. By \cref{lem:conjugation-groups-normal},
\[
\ConjugationGroup_{j+1}\trianglelefteq\ConjugationGroup_j.
\]
Therefore
\[
\ConjugationGroup_{j+1}\leq O_\dimension(\ConjugationGroup_j).
\]
Since $\Paulis\leq\ConjugationGroup_{j+1}$, it follows that
\[
\Paulis\leq O_\dimension(\ConjugationGroup_j).
\]

By \cref{lem:successive-normalization}, every
$V\in\ConjugationGroup_{j-1}$ induces an automorphism of $\ConjugationGroup_j$, so it
preserves $O_\dimension(\ConjugationGroup_j)$. Hence
\[
VP_yV^\dagger\in O_\dimension(\ConjugationGroup_j)
\]
for every $V\in\ConjugationGroup_{j-1}$ and $y\in\PauliLabelSpace$. These elements generate
$\ConjugationGroup_j$, and therefore
\[
\ConjugationGroup_j=O_\dimension(\ConjugationGroup_j).
\]
Thus $\ConjugationGroup_j$ is a $\dimension$-group.

Descending induction proves that
\[
\ConjugationGroup_2,\ldots,\ConjugationGroup_{\level-2}
\]
are finite $\dimension$-groups. Finally,
\[
\ConjugationGroup_1=U\Paulis U^\dagger\cong\Paulis
\]
is automatically a finite $\dimension$-group.
\end{proof}

\begin{remark}
The assumption
\[
\ConjugationGroup_{\level-2}\subseteq\Cliffords_2
\]
implies both
\[
\ConjugationGroup_j\subseteq\Cliffords_{\level-j}
\qquad
1\leq j\leq \level-2
\]
and the nested normal chain
\[
\ConjugationGroup_{\level-2}
\trianglelefteq
\ConjugationGroup_{\level-3}
\trianglelefteq\cdots\trianglelefteq
\ConjugationGroup_2.
\]
The first group $\ConjugationGroup_1$ need not contain $\ConjugationGroup_2$, but it
normalizes $\ConjugationGroup_2$ by
\cref{lem:successive-normalization}.

For example, when $\level=5$,
\[
\ConjugationGroup_3\subseteq\Cliffords_2
\]
implies
\[
\ConjugationGroup_2\subseteq\Cliffords_3,
\qquad
\ConjugationGroup_1\subseteq\Cliffords_4,
\]
and all three groups $\ConjugationGroup_1,\ConjugationGroup_2,\ConjugationGroup_3$ are finite
$\dimension$-groups.
\end{remark}

\begin{theorem}
\label{thm:conditional-higher-gsc}
Let $U\in\Cliffords_\level$, where $\level\geq4$, and suppose that $\ConjugationGroup_{\level-2}\subseteq\Cliffords_2$.
Then $U$ is generalized semi-Clifford.
\end{theorem}

\begin{proof}
By \cref{thm:all-conjugation-groups-d-groups}, every
$\ConjugationGroup_j$ is a finite $\dimension$-group. Let $\Lagrangians$ be the finite set of Lagrangian subspaces of $\PauliLabelSpace$, and
write
\[
\Algebra_L=\lspan\set{P_x:x\in L}.
\]
Since $\ConjugationGroup_{\level-2}\subseteq\Cliffords_2$, conjugation induces an action of
$\ConjugationGroup_{\level-2}$ on $\Lagrangians$. Define
\[
\FixedLagrangians_{\level-2}
:=
\left\{
L\in\Lagrangians:
W\Algebra_LW^\dagger=\Algebra_L
\text{ for every }W\in\ConjugationGroup_{\level-2}
\right\}.
\]
Every nontrivial orbit has cardinality divisible by $\dimension$, so
\[
|\FixedLagrangians_{\level-2}|
\equiv|\Lagrangians|
\equiv1\pmod{\dimension}.
\]
Thus $\FixedLagrangians_{\level-2}$ is nonempty and has cardinality congruent to
$1\pmod{\dimension}$.

For $1\leq j<\level-2$, define
\[
\FixedLagrangians_j
:=
\left\{
L\in\Lagrangians:
V\Algebra_LV^\dagger=\Algebra_L
\text{ for every }V\in\ConjugationGroup_j
\right\}.
\]
We first observe that
\[
\FixedLagrangians_j\subseteq\FixedLagrangians_{j+1}.
\]
Indeed, if $L\in\FixedLagrangians_j$, then every generator
$W=VP_yV^\dagger$ of $\ConjugationGroup_{j+1}$ satisfies
\[
\begin{aligned}
W\Algebra_LW^\dagger
&=VP_y(V^\dagger\Algebra_LV)P_y^\dagger V^\dagger\\
&=VP_y\Algebra_LP_y^\dagger V^\dagger\\
&=V\Algebra_LV^\dagger\\
&=\Algebra_L.
\end{aligned}
\]
Here we used $V^\dagger\Algebra_LV=\Algebra_L$, as $L\in\FixedLagrangians_j$ and $\ConjugationGroup_j$ is a group, together with $P_y\Algebra_LP_y^\dagger=\Algebra_L$, since conjugation by a Pauli sends each $P_x$ to a scalar multiple of itself and hence preserves $\Algebra_L$.

Now fix $2\leq j\leq \level-2$. We claim that $\ConjugationGroup_{j-1}$ acts on
$\FixedLagrangians_j$. Let $L\in\FixedLagrangians_j$, $V\in\ConjugationGroup_{j-1}$, and set
\[
\Algebra_V:=V\Algebra_LV^\dagger.
\]
For every Pauli $P_y$, we have
\[
V^\dagger P_yV\in\ConjugationGroup_j,
\]
and hence
\[
\begin{aligned}
P_y\Algebra_VP_y^\dagger
&=V(V^\dagger P_yV)\Algebra_L
  (V^\dagger P_y^\dagger V)V^\dagger\\
&=V\Algebra_LV^\dagger\\
&=\Algebra_V.
\end{aligned}
\]
By \cref{cor:pauli-masa}, $\Algebra_V$ is a Pauli MASA.

Furthermore, for every $W\in\ConjugationGroup_j$,
\[
\begin{aligned}
W\Algebra_VW^\dagger
&=V(V^\dagger WV)\Algebra_L
  (V^\dagger W^\dagger V)V^\dagger\\
&=V\Algebra_LV^\dagger\\
&=\Algebra_V,
\end{aligned}
\]
because $V^\dagger WV\in\ConjugationGroup_j$ by
\cref{lem:successive-normalization}. Thus $\ConjugationGroup_{j-1}$ acts on
$\FixedLagrangians_j$.

The fixed-point set of this action is precisely $\FixedLagrangians_{j-1}$, since
$\FixedLagrangians_{j-1}\subseteq\FixedLagrangians_j$. Because $\ConjugationGroup_{j-1}$ is a $\dimension$-group,
\[
|\FixedLagrangians_{j-1}|\equiv|\FixedLagrangians_j|\pmod{\dimension}.
\]
Starting from $\FixedLagrangians_{\level-2}$ and descending inductively gives
\[
|\FixedLagrangians_1|\equiv1\pmod{\dimension}.
\]
Therefore there exists $L\in\Lagrangians$ such that
\[
G_x\Algebra_LG_x^\dagger=\Algebra_L
\qquad\text{for every }x\in\PauliLabelSpace.
\]

Finally, set
\[
\Algebra:=U^\dagger\Algebra_LU.
\]
For every Pauli $P_x$,
\[
\begin{aligned}
P_x\Algebra P_x^\dagger
&=U^\dagger G_x\Algebra_LG_x^\dagger U\\
&=U^\dagger\Algebra_LU\\
&=\Algebra.
\end{aligned}
\]
By \cref{cor:pauli-masa}, $\Algebra$ is a Pauli MASA. Thus
$\Algebra=\Algebra_S$ for some $S\in\Lagrangians$, and
\[
U\Algebra_SU^\dagger=\Algebra_L.
\]
Therefore $U$ is generalized semi-Clifford.
\end{proof}

\section*{Acknowledgements} Subramanian acknowledges support from the Royal Society through a University Research Fellowship. 

\section*{AI Disclosure}

We used OpenAI Codex to assist with mathematical exploration, including developing proof strategies, checking intermediate algebraic arguments, and, alongside Claude Code, to identify relevant literature and draft text. The tools materially influenced the development and exposition of arguments throughout the manuscript. All AI-assisted material was reviewed and substantially edited by the authors. The authors independently verified the mathematical claims, proofs, calculations, and references and take full responsibility for the correctness and originality of the manuscript.

\bibliographystyle{alpha}
\bibliography{refs}

\appendix
\crefalias{section}{appendix}

\section{Deferred proofs}
\label{sec:deferred-proofs}
\begin{lemma}[\cref{lem:lagrangian-subspaces}, restated]
The number of Lagrangian subspaces of $\PauliLabelSpace$ is
\[
\prod_{j=1}^n(d^j+1).
\]
\end{lemma}
\begin{proof}
\cite{Gross2006} gives the number of isotropic subspaces of dimension $m$ in $\Field_\dimension^{2\qudits}$ as 
\[
\prod_{i=0}^{m-1}\frac{\dimension^{2(\qudits-i)}-1}{\dimension^{m-i}-1}\enspace.
\]
Since Lagrangians correspond to the subspaces with $m=\qudits$, the number of Lagrangian subspaces of $\PauliLabelSpace$ is 

\begin{equation*}
\prod_{i=0}^{\qudits-1} (\dimension^{n-i}+1) =
\prod_{j=1}^{\qudits} (\dimension^{j}+1)\qedhere
\end{equation*}
\end{proof}

\begin{lemma}[\cref{lem:pauli-invariance}, restated.]
Let $\MatrixSubspace \subseteq M_{\Dimension
}(\Complex)$ be a linear subspace satisfying
\[
P_x \MatrixSubspace P_x^\dagger=\MatrixSubspace
\]
for all $x \in \Field_\dimension^{2\qudits}$. Then $\MatrixSubspace = \lspan_\Complex\Subset$ with $\Subset \subseteq \Paulis$. 
\end{lemma}
\begin{proof}
For any $x\in \PauliLabelSpace$, the projector $\Pi_x$ onto $\lspan_\Complex\set{P_x}$ can be written as
\[
\Pi_x(M):=\frac{1}{\abs{\PauliLabelSpace}}\sum_{y\in \PauliLabelSpace}
\RootOfUnity^{[x,y]}P_y M P_y^\dagger.
\]
Indeed, 
\[
\begin{aligned}
\Pi_x(P_z)
=\frac{1}{\dimension^{2\qudits}}\sum_{y\in \PauliLabelSpace}\RootOfUnity^{[x,y]}P_yP_zP_y^\dagger  &=\frac{1}{\dimension^{2\qudits}}\sum_{y\in \PauliLabelSpace}\RootOfUnity^{[y,z - x]}P_z\\
&=\begin{cases}
P_z &\text{ if } z=x,\\
0 & \text{otherwise,}
\end{cases}
\end{aligned}
\]

Since $\MatrixSubspace$ is invariant under every Pauli conjugation,
$W_y:=P_y W P_y^\dagger\in \MatrixSubspace$ for any $W\in \MatrixSubspace$, and thus $\Pi_x(W) \in \MatrixSubspace$ by closure under linear combinations. Consequently, if $\Pi_x(W) \propto P_x$ is nonzero for some $W\in \MatrixSubspace$, then $P_x\in \MatrixSubspace$.

Now, let
\[
S=\set{x: \exists W \in \MatrixSubspace \text{ such that } \Pi_x(W)\neq 0};
\] then 
\[
\MatrixSubspace=\lspan\set{P_x : x \in S},
\]
as we prove next. The containment $\lspan\set{P_x : x \in S}\subseteq \MatrixSubspace$ follows from $P_x \in \MatrixSubspace$ when $\Pi_x(W) \neq 0$, as shown above.

Since every operator admits a Pauli decomposition,
\[
W=\sum_{x\in T} \alpha_x P_x 
\]
for some $T\subset \PauliLabelSpace$. Since $\Pi_x(W)\neq 0$ if $x\in T$, we have $T \subseteq S$ and thus $W\in \lspan\set{P_x : x \in S}$. As $W\in\MatrixSubspace$ was arbitrary, $\MatrixSubspace \subseteq \lspan\set{P_x : x \in S}$.
\end{proof}

\end{document}